\documentclass[letterpaper]{article} % DO NOT CHANGE THIS
\usepackage{aaai24}  % DO NOT CHANGE THIS
\usepackage{times}  % DO NOT CHANGE THIS
\usepackage{helvet}  % DO NOT CHANGE THIS
\usepackage{courier}  % DO NOT CHANGE THIS
\usepackage[hyphens]{url}  % DO NOT CHANGE THIS
\usepackage{graphicx} % DO NOT CHANGE THIS
\usepackage{natbib}  % DO NOT CHANGE THIS AND DO NOT ADD ANY OPTIONS TO IT
\usepackage{caption} % DO NOT CHANGE THIS AND DO NOT ADD ANY OPTIONS TO IT
\usepackage{algorithm}
\usepackage{algorithmic}

\usepackage{booktabs, amsfonts, nicefrac, microtype, amsmath} 

\usepackage{color, amssymb, amsthm, bigints, graphicx, subcaption, pifont, float}
\newcommand{\cmark}{\ding{51}}%
\newcommand{\xmark}{\ding{55}}%
\theoremstyle{definition}

\newtheorem{proposition}{Proposition}
\newtheorem{assumption}{Assumption}

\newtheorem{remark}{Remark}

\usepackage{newfloat}
\usepackage{listings}
\DeclareCaptionStyle{ruled}{labelfont=normalfont,labelsep=colon,strut=off} % DO NOT CHANGE THIS
\floatstyle{ruled}
\newfloat{listing}{tb}{lst}{}
\floatname{listing}{Listing}
\title{Meta-Learning-Based Adaptive Stability Certificates for Dynamical Systems}
\author{
    Amit Jena\textsuperscript{\rm 1},
    Dileep Kalathil\textsuperscript{\rm 1},
    Le Xie\textsuperscript{\rm 1}
}
\affiliations{
    \textsuperscript{\rm 1}Department of Electrical and Computer Engineering, Texas A\&M University, USA \\
    amit\_jena@tamu.edu, dileep.kalathil@tamu.edu, le.xie@tamu.edu
}

\usepackage{bibentry}
\begin{document}

\maketitle

\begin{abstract}
This paper addresses the problem of Neural Network (NN) based adaptive stability certification in a dynamical system. The state-of-the-art methods, such as Neural Lyapunov Functions (NLFs), use NN-based formulations to assess the stability of a non-linear dynamical system and compute a Region of Attraction (ROA) in the state space. However, under parametric uncertainty, if the values of system parameters vary over time, the NLF methods fail to adapt to such changes and may lead to conservative stability assessment performance. We circumvent this issue by integrating Model Agnostic Meta-learning (MAML) with NLFs and propose meta-NLFs. In this process, we train a meta-function that adapts to any parametric shifts and updates into an NLF for the system with new test-time parameter values. We demonstrate the stability assessment performance of meta-NLFs on some standard benchmark autonomous dynamical systems.
\end{abstract}

%%%%%%%%%% 1:Introduction %%%%%%%%%%%%%%%%%%

\section{Introduction}
Stability assessment of non-linear systems and ensuring their safe and reliable operation are of paramount importance in any real-world engineering system.  While learning-based control schemes have received  a lot of attention recently, the lack of stability guarantees is a fundamental issue that prevents their wide-scale deployment in the real world.   One standard approach to estimate the stability region of a general nonlinear system is to first find a Lyapunov function for the system and characterize its region of attraction (ROA) as the stability region \cite{khalil2015nonlinear}. A closed-loop system is stable in the sense of Lyapunov if the system trajectory converges to the origin as long as the initial condition is inside the ROA. The sum-of-squares approach is one popular method for finding a Lyapunov function for a dynamical system \cite{parrilo2000structured, henrion2005positive,jarvis2003some,topcu2009robust, topcu2008local}. However, sum-of-squares approaches typically do not scale well to large systems since a large number of semidefinite programs need to be solved for the sum-of-squares decomposition of polynomial systems even with a few states \cite{parrilo2000structured}.  Another approach is to employ local linearizations and use quadratic  approximations to find Lyapunov functions. However, this approach is typically conservative in the sense that stability can only be certified in a small vicinity of an equilibrium point of a nonlinear system, which may be insufficient to cover the normal range of operation in several practical applications \cite{cheng2003quadratic,huang2021neural}.

Recently, a line of works has used learning-based approaches that make use of the function approximation capability of a neural network for finding the Lyapunov function  \cite{kolter2019learning, chang2019neural}. The key idea is to use supervised learning to find a neural Lyapunov function (NLF) by  minimizing a loss function that captures the  Lyapunov constraints. The NLF approaches have shown impressive empirical performance in estimating the stability region for nontrivial nonlinear systems. However, training an NLF requires a large number of data samples and gradient update steps, which makes the real-time training infeasible.  So, a typical fix is to use the  system model or pre-collected data from the real-world system and perform offline training. The trained NLF  can then be used for the stability estimation of the real-world system. However, this approach will fail if the real-world system dynamics is different from the model used for training the NLF. At the same time,  the real-world system model can be different from the model estimated from the collected data due to various reasons, such as estimation error and changes in the system parameters over time. Repeating the training procedure every time whenever there is such a parametric mismatch turns impractical due to the unavailability of necessary data samples and the need to get a quick stability assessment. Thus, \textit{learning a neural Lyapunov function for a real-world system using only a small number of data samples and through a few gradient updates}, remains an open problem.

In this paper, we address the problem of stability assessment of a closed-loop system, focusing on the setting where the real-world system parameters are different from the offline model parameters. Our goal is to learn a NLF that : (i) quickly adapts to a new system instance under parametric shifts, (ii) and does so with low adaptation-time data requirements. We follow a meta-learning procedure, and in particular, the framework of Model Agnostic Meta-learning (MAML) \cite{finn2017model}, to learn such a meta-neural Lyapunov function (meta-NLF).  We compare the performance of the meta-NLF approach with other baseline methods on various benchmark control systems to demonstrate the efficacy of our method. Our codes and appendix are available at \url{https://github.com/amitjena1992/Meta-NLF}.
%%%%%%%%%%%%%%%%%%%%%%%%%%%%%%%%%%%%

%%%%%%%%%%%%%%%%%% 2: Related work %%%%%%%%%%%%%%%%

\section{Related Work}

Computing Lyapunov functions is one of the classical techniques for estimating the stability of non-linear systems. Various methods try to learn Lyapunov functions as polynomials \cite{kapinski2014simulation, ravanbakhsh2016robust}, support vectors \cite{khansari2014learning}, Gaussian processes \cite{umlauft2018uncertainty} and temporal logic formulae \cite{jha2017telex}. A few widely-adopted ones among such methods are Sum of Squares (SOS) based polynomial approaches \cite{parrilo2000structured}, and quadratic Lyapunov functions \cite{cheng2003quadratic}.

To overcome a few existing shortcomings of model-based approaches, a set of recent works \cite{richards2018lyapunov, kolter2019learning, chang2019neural, taylor2019episodic, choi2020reinforcement, mehrjou2020neural, jin2020neural, dai2021lyapunov, dawson2022safe} employ neural networks to approximate Lyapunov functions. The core of these works is to train a neural network to optimize a loss function that captures the Lyapunov constraints. These methods require either the dynamical system's closed-form expression or offline trajectory data thereof \cite{boffi2021learning}. Such methods usually perform well but remain susceptible to parameter shifts in dynamical systems.

Meta-learning provides a framework where a machine learning model gathers experience over multiple training tasks and leverages this experience to better its future learning performance on similar tasks \cite{schmidhuber1997shifting, meta-bengio1990learning, meta-thrun1998learning, meta-hochreiter2001learning, meta-younger2001meta}. In particular,  Model Agnostic Meta-learning (MAML) \cite{finn2017model} is a well-recognized method that is used in many areas, including computer vision \cite{meta-cv-achille2019task2vec}, natural language processing \cite{meta-nlp-huang2018natural}, reinforcement learning  \cite{meta-rl-nagabandi2018learning} and system identification and control \cite{meta-idc-shi2021meta}. However, to the best of our knowledge, a meta-learning-based stability assessment of dynamical systems has not been addressed before. 
%%%%%%%%%%%%%%%%%%%%%%%%%%%%%%%%%%%%%%%%

%%%%%%%%%%%%%%%%%%%%%%% 3: Preliminaries %%%%%%%%

\section{Preliminaries}
We consider a continuous-time (closed-loop) autonomous dynamical system of the form
\begin{equation}
\label{eq:dyn_def}
    \dot x = f_{\vartheta}(x),
\end{equation}
where the state  $x \in X \subset \mathbb{R}^d$ is fully observed, and the  dynamics is represented by the continuous function $f_{\vartheta}:\mathbb{R}^d \rightarrow \mathbb{R}^d$ where $\vartheta \in \mathbb{R}^\mathfrak{d}$ is the parameter that characterizes the dynamics. For any given initial condition $x_{0}$, we assume that this   dynamics gives a unique trajectory $x(t; x_{0}), t \geq 0$, with $x(0; x_{0}) = x_{0}$. In the following, we simply represent the trajectory as $x(t)$, making the dependence on the initial condition implicit. 

A closed-loop system  is said to be stable  (in the sense of Lyapunov) at the origin ($x \equiv 0$) if for any $\epsilon > 0$; there exists a $\delta > 0$ such that whenever $\|x(0)\| < \delta \implies \|x(t)\| <\epsilon,~~\forall t >0$. The system is said to be asymptotically stable at the origin if it is  stable and there exists a $\delta > 0$ such that  $\lim_{t\rightarrow \infty}\|x(t)\| = 0$ for any initial condition $\|x(0)\| < \delta$. One traditional way to certify the asymptotic stability of a system is through  a Lyapunov function. 

A continuously differentiable function $V: D \rightarrow \mathbb{R}$, $D \subset X$, is called a (strict) Lyapunov function for the system  \eqref{eq:dyn_def} if it satisfies the following conditions: 
\begin{subequations}
\label{eq:lyapunov-condition}
\begin{align}
\label{eq:lyapunov-condition-1}
    V(0) &= 0, \\
    \label{eq:lyapunov-condition-2}
    V({x}) &> 0,~\forall x \in D \setminus \{0\}, \\
    \label{eq:lyapunov-condition-3}
    \dot{V}({x}) &= \nabla V({x})^\top {f}({x})<0, ~\forall x \in D \setminus \{0\}.
    \end{align}
\end{subequations}

If $V$ is the Lyapunov function for the system  \eqref{eq:dyn_def}, then the system is asymptotically stable, \cite{khalil2015nonlinear}. The region $D$ is called the valid region which contains the ROA in the form of the largest level set of $V$.

In general, sum-of-squares approaches \cite{parrilo2000structured,henrion2005positive,jarvis2003some,topcu2009robust,topcu2008local} can be used to assess the stability of  nonlinear systems. Alternatively, local linearizations of  the dynamics  may be used to compute a quadratic Lyapunov function \cite{chiang1989study}. However, these approaches often fail to find a meaningful Lyapunov function and ROA for large-scale, high-dimensional,  networked systems because of the following challenges. Firstly, sum-of-squares approaches do not scale well computationally and will quickly become intractable when the system dimension increases \cite{parrilo2000structured,henrion2005positive,jarvis2003some,topcu2009robust,topcu2008local}. Secondly, sum-of-squares and quadratic approximation-based approaches are typically very conservative in their estimation of ROA even for small systems \cite{chang2019neural}. This may lead to the design of conservative controllers and sub-optimal system operation. 

In order to overcome these  challenges, some recent works have exploited the data-based  function approximation capability of a neural network to \textit{learn} Lyapunov functions for nonlinear systems  \cite{kolter2019learning, chang2019neural}. The goal is to learn an NLF $V_{\theta}$ where $\theta$ is the parameter of the neural network, which represents the Lyapunov function.  In particular, \citet{chang2019neural} has defined the Lyapunov loss function as 
\begin{align}
    \label{eq:lyapunov-risk1} 
    L(\theta) &= \mathbb{E}_{(x,y) \sim \rho_\vartheta} [\ell(\theta, (x,y))],~ \text{where} \\
    \label{eq:lyapunov-risk2} 
    \ell(\theta, (x,y)) &=  [  \max (0, -V_\theta(x)) \nonumber \\
    &\hspace{1cm} + \max(0, \nabla V_\theta(x)^\top y) +V_\theta^2(0)], 
\end{align}
with $y = f_{\vartheta}(x)$ and $\rho_{\vartheta}(x)$ being a joint distribution over $(x, y)$.  The loss function $\ell(\cdot, \cdot)$ is defined to  translate each condition in \eqref{eq:lyapunov-condition} to an equivalent loss. In particular, the first term is to ensure the condition \eqref{eq:lyapunov-condition-2} that the $V_{\theta}$ is non-negative, the second term is  to ensure the condition \eqref{eq:lyapunov-condition-3} that  the Lie derivative of $V_{\theta}$ is negative and the third term enforces the condition \eqref{eq:lyapunov-condition-1} that $V_{\theta}(0)$ is zero.  The learning algorithm then finds a parameter by minimizing the empirical loss function 
\begin{align}
    \label{eq:lyapunov-empirical-loss1} 
    \widehat{L}(\theta, S) &= \frac{1}{|S|} \sum_{z \in S} \ell(V_\theta, z),
\end{align}
where each $z = (x,y) \in S$ is generated according to the distribution $\rho_{\vartheta}$ with $y = f_{\vartheta}(x)$ . The NLF approaches have shown impressive empirical performance in estimating the stability region for nontrivial nonlinear systems  \cite{kolter2019learning, chang2019neural, huang2021neural}.
%%%%%%%%%%%%%%%%%%%% 4: Formulation %%%%%%%%%%%%%% 
\section{Meta-Neural Lyapunov Function}

The NLF approach follows a standard supervised learning method that requires a large number of training data samples from a fixed system and a large number of stochastic gradient descent-based parameter updates to  learn the Lyapunov function. So, the NLF training can only be done offline, and the trained NLF can then be used to certify the stability of the real-world system. However, the operating conditions of  many real-world engineering systems can vary over time, which will also change their closed-loop system behavior. For example, in electricity distribution systems,  the load pattern, topology, and line impedance can change due to the presence of rooftop photovoltaics (PVs), electric vehicles, battery storage, power electronics devices, and other uncertainty-inducing components. This presents two main challenges for the NLF approach to be used for certifying the stability  of the real-world systems: (i) data used for the offline training of the NLF may not represent the behavior of the real-world system, and hence the pre-trained NLF may not give the correct stability certificate for this system, (ii) it may not be possible to collect enough data from the real-world system and perform a large number of gradient updates to learn a new NLF for this system.  We propose to overcome these challenges with a meta-learning approach. In particular, through the model-agnostic meta-learning (MAML)  \cite{finn2017model}, we seek to formulate a meta-NLF capable of adapting to the real-world system with a small number of training samples and gradient steps. In what follows, we refer to the meta-NLF and its task-adapted NLF as $V_\theta$ and $V_{\theta^{'}}$ respectively where $\theta^{'}$ results from a K-step adaptation of the meta-parameter $\theta$ on a given dynamical system.

\begin{figure*}[t]
     \begin{subfigure}[b]{0.24\textwidth}
         \centering
         \includegraphics[width=\textwidth]{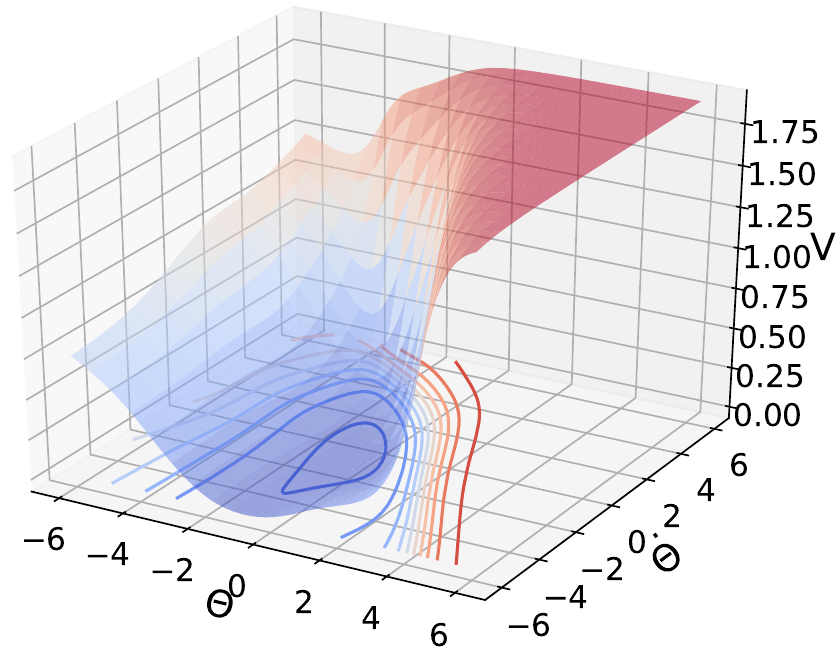}
        \caption{}
     \end{subfigure}
     \hfill
     \begin{subfigure}[b]{0.24\textwidth}
         \centering
         \includegraphics[width=\textwidth]{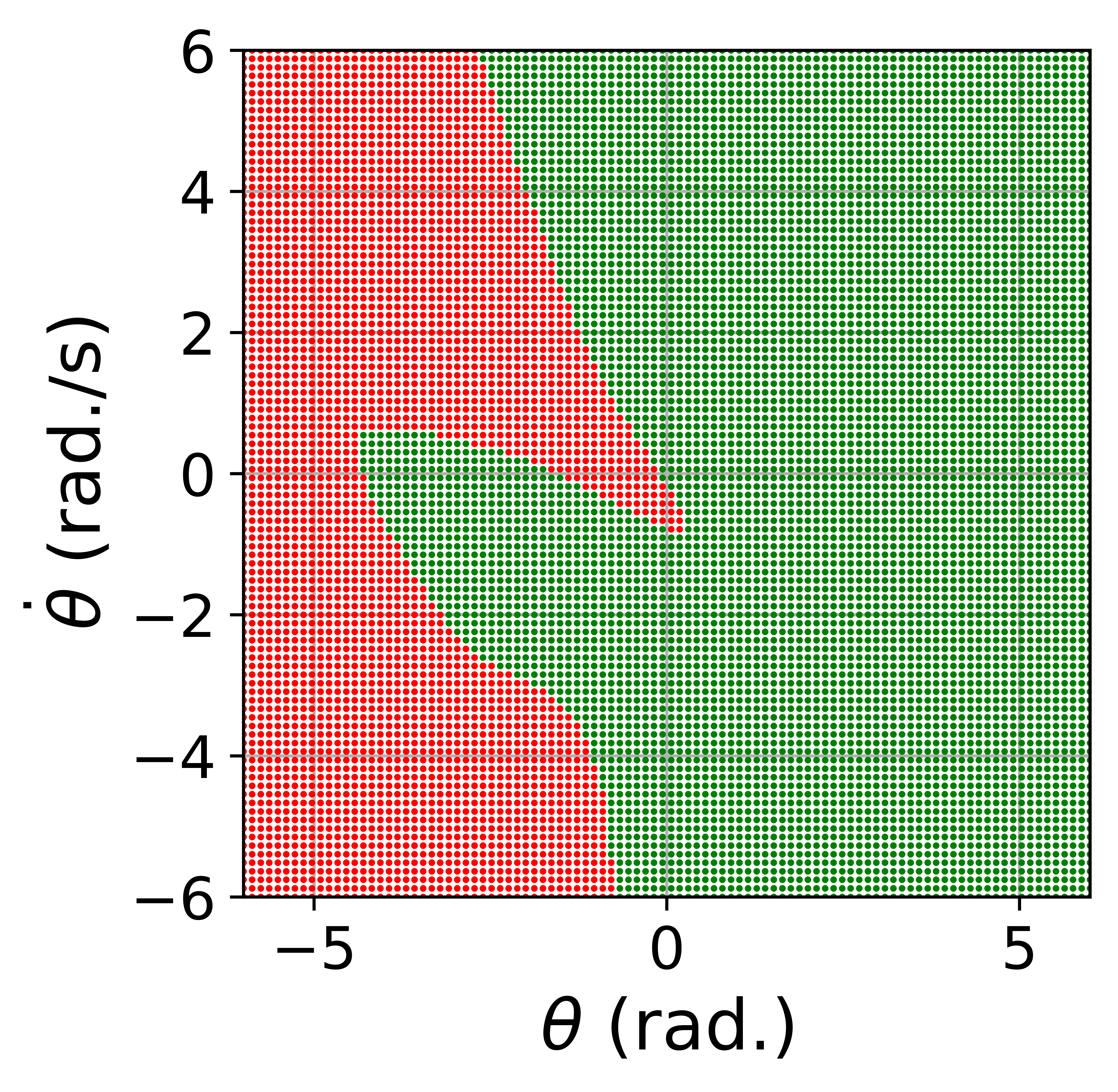}
        \caption{}
        \label{fig:transition-ROA-meta}
     \end{subfigure}
     \hfill
     \begin{subfigure}[b]{0.24\textwidth}
         \centering
         \includegraphics[width=\textwidth]{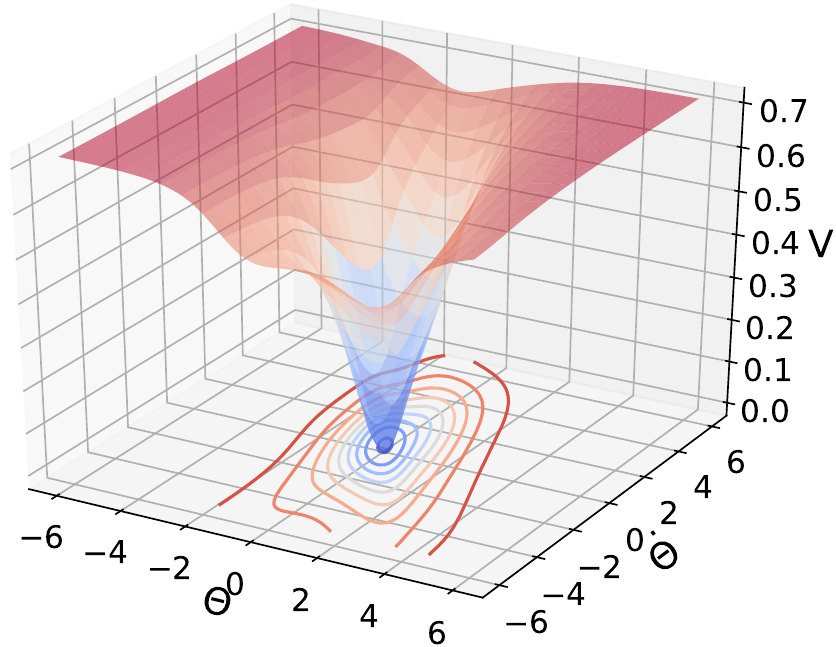}
        \caption{}
     \end{subfigure}
     \hfill
     \begin{subfigure}[b]{0.24\textwidth}
         \centering
         \includegraphics[width=\textwidth]{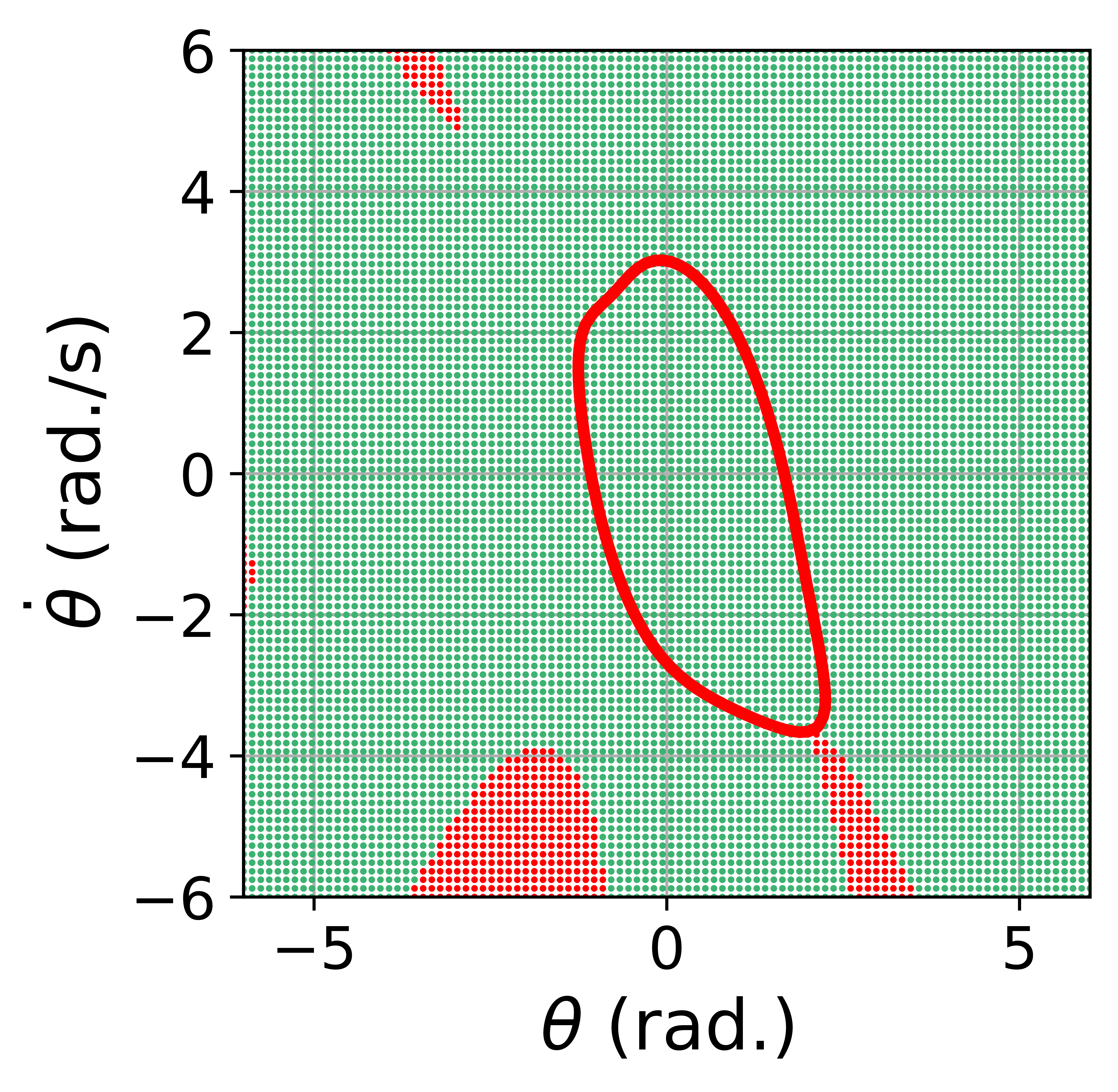}
        \caption{}
        \label{fig:transition-ROA-task}
     \end{subfigure}
      % \label{}
     \caption{\textit{Plots describing the transition of meta-NLF into a task-adapted NLF}. All figures correspond to the Inverted Pendulum setup with stochastic $l$. A well-calibrated meta-training, results in the meta-NLF shown in (a). In (b), the area in $D$ where the meta-NLF satisfy the Lyapunov constraints is colored in green, and the other region is colored in red. A test-time system and a few samples thereof lead to the task-adapted NLF in (c). The valid and invalid regions (Similar to Figure 1b.) for the task-adapted NLF are presented in (d), along with an ROA shown as a contour in red. Although the meta-NLF trains with the Lyapunov loss across a set of system instances, it itself doesn't serve as a Lyapunov function on any test-time system which is evident from Figure. 1b, where sketching an ROA isn't possible. However, a one-step gradient update leads to the task-adapted NLF whose stability assessment performance considerably improves than that of the meta-NLF.}
     \label{fig:transition-meta-to-task}
    \end{figure*}

\subsection{Verification of Meta-NLF}
\label{subsubsec:Ver-meta-nlf}
For a Lyapunov function to be valid, constraints \eqref{eq:lyapunov-condition-1}, \eqref{eq:lyapunov-condition-2} and \eqref{eq:lyapunov-condition-3} must satisfy everywhere in $D$.  But, the minimization of the Lyapunov loss function (as stated in \eqref{eq:lyapunov-risk2}) doesn't automatically translate to the satisfaction of Lyapunov constraints in $D$. Thus, a verification process such as an SMT solver \cite{chang2019neural} or the Lipschitz technique \cite{richards2018lyapunov}
becomes necessary. Using a logic formula that is made up of negation of all Lyapunov constraints, the SMT solver filters out state vectors 
that satisfy the logic formula \textit{or violate the Lyapunov conditions}. However, from an implementation standpoint,  this method doesn't scale well to larger systems, as reported in \cite{jena2022distributed}. Hence, we adopt the Lipschitz 
method introduced in \cite{richards2018lyapunov}. The idea here is that in order to enforce $\dot V (x) < 0$ for  $x \in D$, the tightened condition
$\dot V (x) < -K\tau = -\epsilon_2$ for finitely many samples in $D$ serves as a sufficient condition, where the time derivative of the Lyapunov function $\dot V(x)$ is assumed to be a $K-$Lipschitz function (with respect to 1-norm) and $\tau >0$ signifies how densely the samples cover $D$. We make use of this result to ensure
the time derivative of a task-adapted NLF stays valid across $D$ once it meets the tightened conditions at each node in a fine-grained discretized grid in $D$. A detailed proof of this result can be found in \cite{berkenkamp2017safe}. For positive definiteness, we argue along similar lines that for a Lipschitz task-adapted NLF $V_{\theta^{'}}$, if the tightened condition $V_{\theta^{'}}(x) > \epsilon_1 >0$ holds for finitely many sample points in $D$, then $V_{\theta^{'}}$ is positive definite throughout $D$. Proof of this result is provided in the appendix. Finally, we use the same technique as in \cite{chang2019neural}, and subtract a bias term  $V_{\theta^{'}}(0)$ from a task-adpated NLF to ensure the satisfaction of condition \eqref{eq:lyapunov-condition-1}. For meticulously-optimized meta-NLFs, this bias term attains a small value on any training-time or test-time system, thus not affecting the stability performance significantly.  

Next, we modify the loss function presented in \eqref{eq:lyapunov-risk2} to account for the new tightened versions of \eqref{eq:lyapunov-condition-2} and \eqref{eq:lyapunov-condition-3} as the following:
\begin{eqnarray}
\label{eq:tightened-lyap-cond}
     & \ell(\theta, (x,y)) =  [  \max (0, \epsilon_1-V_\theta(x)) \nonumber \\
    &\hspace{1cm} + \max(0, \epsilon_2+\nabla V_\theta(x)^\top f_\vartheta(x)) +V_\theta^2(0)].
\end{eqnarray}
Here, we consider $\epsilon_1$ and $\epsilon_2$ tunable hyperparameters. In addition to including these constants in the loss function, we also rigorously check the tightened conditions on a discretized grid in $D$ once a trained meta-NLF adapts to a test-time system. We provide a visual presentation of this validation in Figure \ref{fig:transition-ROA-task} where if the task-adapted NLF satisfies the tightened Lyapunov constraints at a node in the grid, the node is marked green and marked red otherwise.

Next, we provide the Lipschitz assumptions on the dynamics $f(x)$, meta-NLF $V_\theta(x)$, and gradient of meta-NLF $\nabla V_\theta$. It's straightforward to prove that Lipschitzness of $\nabla V_\theta(x)$ and $f(x)$ result in $\dot V_\theta(x) = \nabla V_\theta(x)^T f(x)$ being Lipschitz. 
\begin{assumption}
\label{assum:lipschitz_}
For any $x_1, x_2 \in X$, the following satisfy for any $\theta \in \Theta$:
\begin{subequations}
\begin{align}
    |V_{\theta}(x_1) - V_{\theta}(x_2)| &\leq K_{V} \|x_1-x_2\| \label{eq:-lip-lyap-a},\\
    \|\nabla V_{\theta}(x_1) - \nabla V_{\theta}(x_2)\|_1 &\leq K_{\nabla V} \|x_1- x_2\|_1 \label{eq:-lip-grad-b}, \\
    \|f(x_1) - f(x_2)\|_1 & \leq K_{f} \|x_1- x_2\|_1 \label{eq:-lip-fun-c}
\end{align}
\end{subequations}
where $K_V$, $K_{\nabla V}$ and $K_f$ are positive Lipschitz constants and $\|.\|_1$ is the standard $\ell_1$ norm.
\end{assumption}
According to the above assumption, the Lipschitzness is independent of the choice of $\theta$. Thus, it's easy to verify that the task-adapted NLF $V_{\theta^{'}}(x)$ and its gradient will naturally inherit the Lipschitz property from the meta-NLF.

\subsection{Problem Formulation}
\label{subsec:prob_form}
The MAML framework learns a meta-parameter using the data from multiple tasks. In our context, a task $i$ corresponds to learning the NLF for the system $\dot{x} = f_{\vartheta_{i}}(x)$ with parameter $\vartheta_{i}$. For any $i$-th task, we assume that the algorithm has access to a dataset $S_{i}$ consisting of $m_i$ mini data-batches $S_{i} = \{(S^{\mathrm{tr}}_{i,j}, S^{\mathrm{te}}_{i,j})\}_{j=1}^{m_i}$. Each mini data-batch $(S^{\mathrm{tr}}_{i,j}, S^{\mathrm{te}}_{i,j})$ comprises of $K$ and $J$ number of samples, where each sample $z \in S_{i}$ is $z = (x, y) \sim \rho_{i}$ with $y = f_{\vartheta_{i}}(x)$. The task-specific loss function $L_{i}$ is defined as $L_{i}(\theta)= \mathbb{E}_{z \sim \rho_{i}}[\ell(\theta,z)]$, where $\ell(\cdot, \cdot)$ is as defined in \eqref{eq:tightened-lyap-cond}. The empirical loss function $\widehat{L}_{i}$ is then defined as in \eqref{eq:lyapunov-empirical-loss1}. During the training, the MAML algorithm uses the data from $n$ tasks (containing a total of $\sum_{i=1}^n m_i= m$ mini data-batches). The goal is to learn the meta-parameter that can quickly adapt to a new task  by using $k$-step of stochastic gradient descent  with a batch of size $K$. 

To formally introduce this problem,  define the function $\hat{\mathcal{L}}_{i,j}(\theta)$, which captures the performance of the meta-parameter $\theta$ for task $i$ once it is updated by a single step of gradient descent on $S_{i, j}= (S_{i, j}^\text{tr}, S_{i, j}^\text{te})$,
\begin{align}
    \label{eq:meta-task-loss-1}
    \hat{\mathcal{L}}_{i,j}(\theta) = \hat L_{i}(\theta - \alpha \nabla \hat{L}_{i}(\theta, S_{i,j}^\text{tr}), S_{i,j}^\text{te}).
\end{align}
Next, for any $i$-th task, we define $\mathcal{L}_{i}(\theta) = \mathbb{E}_{S_{i,j} \sim \rho_i^{K+J}}[\hat{\mathcal{L}}_{i,j}(\theta)]$.
The MAML problem is to find the optimal meta-parameter that performs well after one step of adaptation on each task within a given set of tasks. This is posed as the optimization problem  $ \min_{\theta} \mathcal{L}(\theta)$, where $\mathcal{L}(\theta) = \mathbb{E}_{i \sim [n]}[\mathcal{L}_{i}(\theta)]$. The MAML algorithm then solves the empirical problem  $ \min_{\theta} \widehat{\mathcal{L}}(\theta, S)$ where $\widehat{\mathcal{L}}(\theta, S) = \frac{1}{m} \sum^{n}_{i=1} \sum_{j =1}^{m_i} \widehat{\mathcal{L}}_{i, j}(\theta)$

\begin{algorithm}[t]
    \begin{algorithmic}[1]
\caption{ Meta-NLF Training Algorithm}
\label{algor:MAML-NLF}
    \STATE {\textbf{Input:}} $n$ closed-loop systems $\dot x = f_{\vartheta_{i}}(x), 1 \leq i \leq n,$  and the corresponding datasets $\{S_{i}\}_{i=1}^n$,  adaptation and meta step sizes $(\alpha, \tilde \alpha)$, meta-training batch size $P$, total meta-training steps $\mathcal{K}$. 
    \STATE {\textbf{Initialize:}} Meta-parameter $\theta^{0}$ 
   \FOR{$k = 0, \hdots, \mathcal{K}-1$}
   \FOR{$p = 1, \hdots, P$}
      \STATE Randomly sample a task index $i_p$ in $\{1, \hdots, n\}$.
      \STATE Randomly select a mini-batch $S_{i_p, j_p} = (S_{i_p, j_p}^\text{tr}, S_{i_p, j_p}^\text{te})$ from all mini-batches of $i_p-$th task.
   \STATE \textbf{Adaptation:} $\theta_{i_p}^{k} \leftarrow \theta^{k} - \alpha \nabla \hat{L}_{i_p}({\theta^{k}}, {S}_{i_p, j_p}^{\text{tr}})$
   \ENDFOR
   \STATE \textbf{Meta-parameter update}: 
   \begin{align*}
       \theta^{k+1} = \theta^{k} - \tilde\alpha \nabla \frac{1}{P} \sum^{P}_{p=1}   \hat{L}_{i_p}({\theta_{i_p}^{k}}, {S}^{\mathrm{te}}_{i_p, j_p})
   \end{align*}
   \ENDFOR
  \STATE {\textbf{Output:}} $\theta_{\mathrm{mnlf}} = \theta^{\mathcal{K}}$ 
\end{algorithmic}
\end{algorithm}

\subsection{Meta-NLF Training and Adaptation}
Our  motivation is to learn a meta-NLF that can be adapted to learn an NLF for the real-world system whose dynamics can be different from the nominal system model   $\dot{x} = f_{\vartheta_{\mathrm{o}}}(x)$ where $\vartheta_{\mathrm{o}}$ is the parameter of the nominal system. For training, we consider multiple tasks by sampling the system parameters $\vartheta_{i} \sim \mathcal{N}(\vartheta_{\mathrm{o}}, \Sigma_{\vartheta})$. For $i-$th task pertaining to parameter $\vartheta_{i}$, we sample the state $x_{i,j}$ from the domain $D$ and get $y_{i,j} = f_{\vartheta_{i}}(x_{i,j})$ to constitute the data sample $z_{i,j} = (x_{i,j}, y_{i,j})$. A collection of such data samples $(z_{i,j})^{(K+J)m_i}_{j=1}$ form the dataset $S_{i} = \{(S^{\mathrm{tr}}_{i,j}, S^{\mathrm{te}}_{i,j})\}_{j=1}^{m_i}$. The training is  similar to the standard MAML training, which we summarize in Algorithm  \ref{algor:MAML-NLF}. The output of the algorithm is a meta-parameter $\theta_{\mathrm{mnlf}}$ that builds the meta-NLF $V_{\theta_{\mathrm{mnlf}}}$.

During testing, we assume that a new system with parameter $\vartheta_{\mathrm{n+1}}$ is realized, and  a dataset $S_{n+1} = (S^{\mathrm{tr}}_{n+1})$ has been made available. We then perform a $k$-step adaptation on this task with the meta-parameter as the initialization as,
\begin{align}
     \theta^{k+1}_{n+1} \leftarrow \theta^{k}_{n+1} - \alpha \nabla \hat{L}(\theta_{n+1}^k, S_{n+1}^\text{tr}), 
\end{align}
with $\theta^{0}_{n+1} = \theta_{\mathrm{mnlf}}$. A set of figures illustrating the transition from a meta-NLF into a task-adapted NLF is shown in Figure \ref{fig:transition-meta-to-task}. These figures and the simulation results of other experiments communicate a central message that a well-trained meta-NLF tries to find an optimal initialization parameter $\theta_\text{mnlf}$ from which less conservative task-adapted NLFs can be computed in $k$ steps. However, the meta-NLF itself doesn't aim to satisfy the Lyapunov conditions strictly, thus being undeployable for any test-time system without needing further test-time gradient updates.

%%%%%%%%%%%%%%%%%%%%%%%%%%%%%

%%%%%%%%%%%%%%%%%%%%%%%%%% 5: Experiments %%%%%%%

\section{Experiments}
\label{sec:experiments}

We demonstrate the performance of meta-NLF with other standard stability assessment methods on various closed-loop dynamical systems following definition \eqref{eq:dyn_def}. Each system is already equipped with a controller, such as an LQR controller or a power system-specific droop controller \cite{yu2015analysis}. Next, for parametric uncertainty, we assume a few system parameters to take random values in a pre-specified range. For example, the length of an inverted pendulum randomly taking a value in $[0.5, 1.0]$ can cause parametric uncertainty. To make our setup as realistic as real-world's, we assume the following:
\begin{itemize}
    \item The meta-distribution function $P_t$ that governs the uncertainty of parameter vector $\vartheta$ is unknown.
    \item The parametric uncertainty is experienced only through a gap between the deterministic parameter vectors $\vartheta_0$ and $\vartheta_{n+1}$ that instantiate the nominal and test-time systems respectively.
\end{itemize}

\subsection{Baseline Methods}
The baseline methods include three non-adaptive Lyapunov functions, namely the Sum of Squares-based Lyapunov Function (SOS-LF), the Quadratic Lyapunov Function (QLF), and the standard Neural Lyapunov Function (NLF). In several numerical experiments, we have observed that when there is a substantial gap between $\vartheta_0$ and $\vartheta_{n+1}$, each non-adaptive method that's been trained on the nominal system results in a very conservative ROA estimate when directly evaluated (without any further gradient updates) on the test-time system. Furthermore, this approach often doesn't yield an ROA at all. Hence, for a meaningful performance comparison purpose, we have given the non-adaptive methods a significant advantage by making the test-time system available to them. We change the notations of the baselines to \textbf{SOS-LF(TS)}, \textbf{QLF(TS)}, and \textbf{NLF(TS)} to indicate their access to the test-time system. Lastly, an adaptive baseline method is created by training an NLF on the nominal system and then naively transferring (with a few gradient updates) to the test-time system. We call this the transfer learning-based NLF \textbf{(T-NLF)}. To summarize, each method experiences the nominal and test-time systems as per the following setting:
\begin{itemize}
    \item QLF(TS), SOS-LF(TS), and NLF(TS) never see the nominal system but only gain explicit access to the test-time system
    \item T-NLF gets full access to the nominal system for neural network training. Then, it gets little access to the test-time system where a very small dataset (e.g. 50 training samples) is available for quick fine-tuning (e.g. 10 gradient step updates).
    \item In meta-NLF, the nominal system is always available. A set of similar system instances are created around the nominal one for meta-training. A fully meta-trained NN is then given access to the test-time system in the same manner as T-NLF.   
\end{itemize}

\subsection{Selection of a Valid Region \textit{D}}
During meta-training, it's important to select a valid region $D$ in which a meta-NLF corresponding to each training system will stay valid. Due to the lack of scalability to larger systems, unlike the original NLF approach \cite{chang2019neural}, we don't use an SMT solver to select $D$. Instead, we follow an alternative procedure to create an appropriate $D$. During the meta-training, we start with a choice of $D:= \{x : \|x\|_2 \leq d \}$ to train $V_\theta$, where the radius $d$ is a tunable hyperparameter. Then, after a full course of meta-training, we check the tightened conditions for each task-adapted NLF $V_{\theta_i}$, $i \in \{1, \hdots, n\}$ on a fine-grained rectangular grid in the same $D$. If this $D$ contains regions where any $V_{\theta_i}$ violates at least one of the tightened Lyapunov conditions, then we shrink $D$ by reducing radius $d$ and repeat the whole procedure of meta-training. We continue this process until all $V_{\theta_i}$s are completely valid in our choice of $D$. Finally, if no such $D$ exists, the algorithm is concluded to fail to obtain a suitable meta-NLF. For T-NLF and NLF(TS), we follow the same series of steps to compute an appropriate $D$. We admit that our procedure of selecting a $D$ is computationally expensive as it involves validating each task-adapted NLFs after the convergence of meta-training. Designing an algorithm that more efficiently selects a $D$ is a part of our future work.

\subsection{Why an Adaptive Lyapunov Function is Needed?}

From a deployment standpoint, among all the above methods (including meta-NLF), QLF(TS) can be computed in the least time owing to the linearization-based simplification of the dynamics. SOS-LF(TS) comes next in terms of the computational time required to generate a Lyapunov function. However, both QLF(TS) and SOS-LF(TS) are typically conservative in computing the ROAs, and the latter needs the dynamics to be polynomial, which is a stringent requirement. While NLF(TS) overcomes the conservative performance issue, it requires a lot of gradient descent steps and a large number of training samples from the test-time system. Given the safety-critical nature of many applications, relying on NLF(TS) will delay the stability assessment process which in turn might lead to a catastrophic outcome, e.g., a city-wide blackout in the context of power systems. Hence, an adaptive Lyapunov function is needed, which doesn't compromise on the performance side and swiftly adapts to any arbitrary test-time system. We have also employed T-NLF for this purpose only to observe the simulation outcomes contrary to our expectations \textit{when a few training samples from the test-time system are available}. For significantly mismatched nominal and test-time systems, T-NLF's ROA falls behind the ROAs of QLF(TS) and SOS-LF(TS). We recall that the first step of T-NLF computation is training an NLF on the nominal system, which necessitates a huge amount of training data and time. However, T-NLF doesn't give a significant benefit in return. Finally, we argue that meta-NLF, by virtue of a systematic meta-training, will facilitate fast-adaptive and less conservative   
stability assessment on any test-time system. Note that to make a fair comparison with meta-NLF, we don't use the SMT solver-based numerical checker in either T-NLF or NLF(TS).

In the end, we evaluate meta-NLF from two perspectives: i) Can meta-NLF produce a comparable or larger ROA than non-adaptive methods that are explicitly computed on a test-time system? (ii) Can meta-NLF adapt to a test-time system better than T-NLF and fetch a less conservative ROA?

\subsection{Simulations}
\begin{figure*}[!ht]
       \begin{subfigure}[b]{\textwidth}
         \centering
         \includegraphics[width=\textwidth]{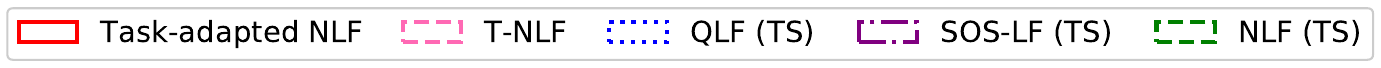}
     \end{subfigure}
     \begin{subfigure}[b]{0.3\textwidth}
         \centering
         \includegraphics[width=\textwidth]{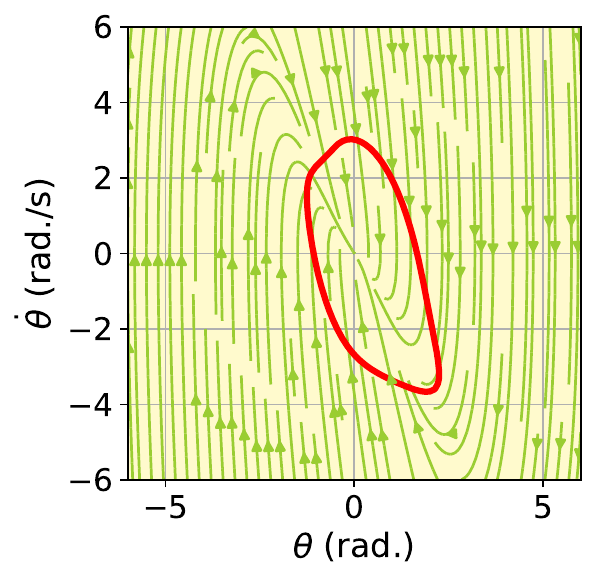}
        \caption{}
        \label{fig:ROAs-IP-1}
     \end{subfigure}
     \hfill
     \begin{subfigure}[b]{0.3\textwidth}
         \centering
         \includegraphics[width=\textwidth]{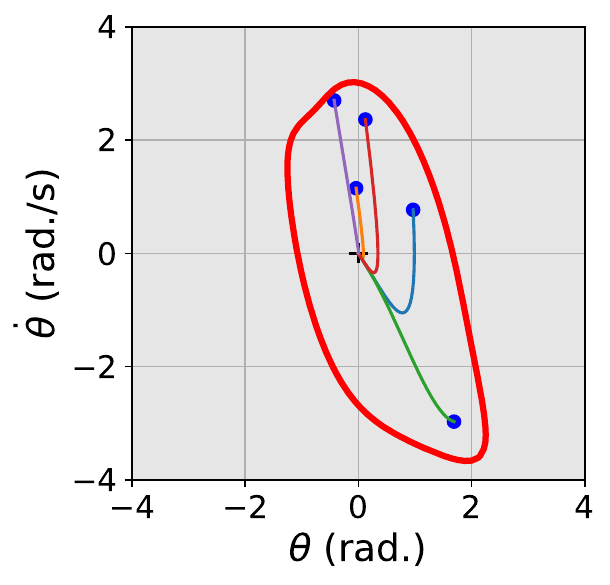}
        \caption{}
        \label{fig:ROAs-IP-2}
     \end{subfigure}
     \hfill
     \begin{subfigure}[b]{0.3\textwidth}
         \centering
         \includegraphics[width=\textwidth]{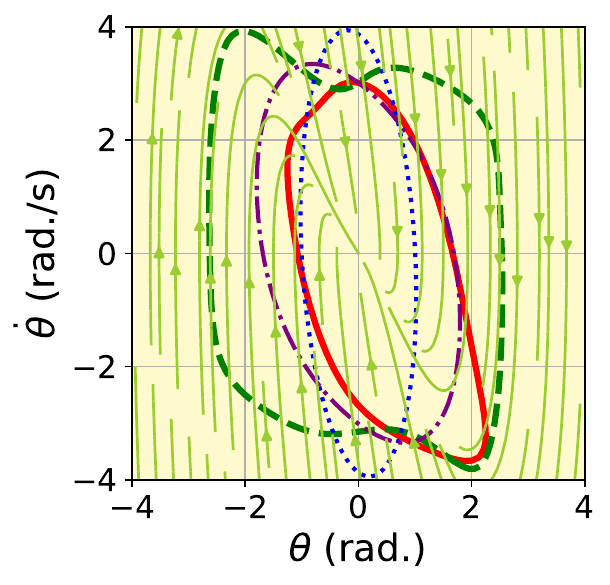}
        \caption{}
        \label{fig:ROAs-IP-3}
     \end{subfigure}
      \label{fig:ROAs-IP}
      \begin{subfigure}[b]{0.3\textwidth}
         \centering
         \includegraphics[width=\textwidth]{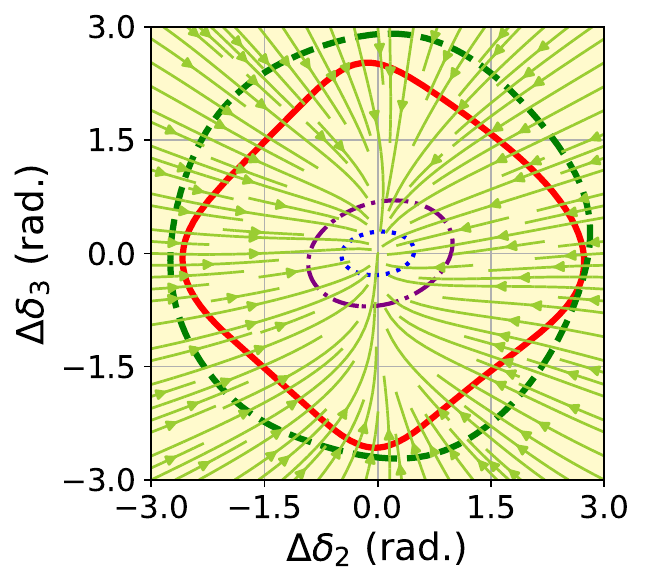}
         \caption{}
        \label{fig:ROAs-3mg-comp-1}
     \end{subfigure}
     \hfill
     \begin{subfigure}[b]{0.3\textwidth}
         \centering
         \includegraphics[width=\textwidth]{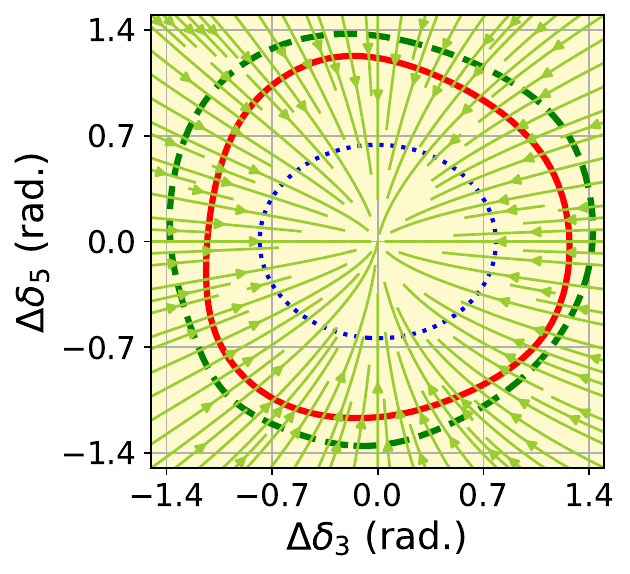}
      \caption{}
        \label{fig:ROAs-5mg-comp-5}
     \end{subfigure}
     \hfill
     \begin{subfigure}[b]{0.3\textwidth}
         \centering
         \includegraphics[width=\textwidth]{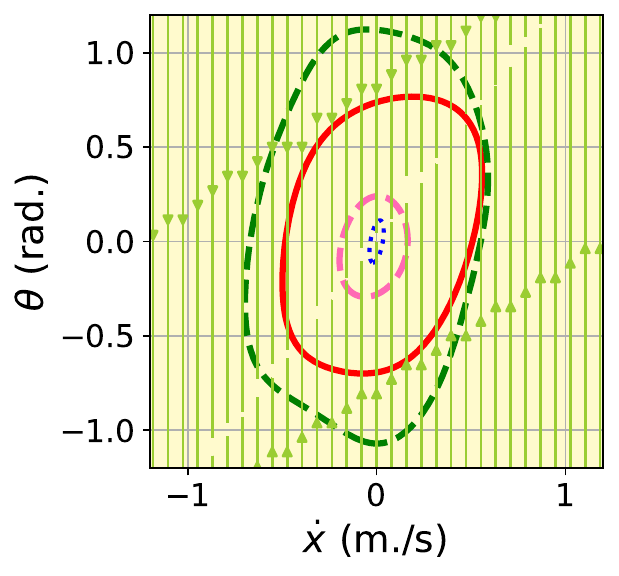}
          \caption{}
        \label{fig:ROAs-cf-comp-2}
     \end{subfigure}
     \caption{\textit{ROA verification and Comparison plots.} In each figure, the phase portrait provides a sense of the system's transient behavior, which serves as a ground truth for understanding each stability assessment method's reliability and conservativeness.
         For a test case of an Inverted Pendulum with stochastic $l$, (a) verifies the ROA of the task-adapted NLF (computed from the meta-NLF on a test-time system) against the phase portrait; (b) validates task-adapted NLF's ROA through time domain simulation; (c) shows task-adapted NLF's performance with other baselines through an ROA comparison. Each subsequent plot concerns a specific test case and compares the ROAs of all methods on it. The test cases are, (d) a Three-microgrid system with stochastic droop constants $(\text{dc}_1, \text{dc}_2)$; (e) a Five-microgrid system with stochastic droop constants $(\text{dc}_1,\hdots, \text{dc}_5)$; (f) Caltech Fan with stochastic $(m, r, d)$. We note that in Figure 2f, the system asymptotically converges to the origin even though the vector fields seem almost parallel to each other. }
    \end{figure*}

We compare meta-NLF's performance with other baselines by matching their ROAs on a test-time system. In each test case, the trained meta-NLF updates into a task-adapted NLF upon facing a test-time system. First, we consider the inverted pendulum environment, where the system parameters are the pendulum's length ($l$) and mass ($m$), acceleration due to gravity ($g$), and coefficient of friction ($b$). The system states are angular displacement from the upright position ($\theta$) and angular velocity ($\dot \theta$). The standard LQR controller closes the control loop so that the system turns ready for stability assessment. Next, we build a test case on this system by assuming $l$ to be stochastic and create the nominal and test-time systems by setting $\vartheta_0 = (l, m, g, b)_0 = (0.5, 0.15, 9,81, 0.1)$ and $\vartheta_{n+1} = (1.2, 0.15, 9.81, 0.1)$. In Figure \ref{fig:ROAs-IP-1}, we overlay the phase portrait of the system with task-adapted NLF's ROA to notice the latter capturing a region where the system tends to gravitate towards the origin. This testifies to the accuracy of our meta-learning-based stability assessment, which we again cross-check through a time domain simulation in Figure \ref{fig:ROAs-IP-2}. Here, the system converges to the origin from random starting points inside the  task-adapted NLF's ROA. Next, we move to Figure \ref{fig:ROAs-IP-3}, where the task-adapted NLF is comparable with QLF(TS) and SOS-LF(TS) but significantly falls behind NLF(TS) in terms of ROA area. T-NLF fails to capture an ROA in this case which suggests the superior adaptation of meta-NLF over the former method. We remark that for small dimensional systems with a few stochastic parameters (e.g., this setting where $l$ is stochastic), meta-NLF's full potential might not be quite evident. We perform additional experiments (see the Appendix) on the same system by assuming more parameters to be stochastic and find results where the task-adapted NLF convincingly beats SOS-LF(TS) and QLF(TS) in performance. In all our experiments, we observe the NLF(TS) to generate the largest ROA among all methods. This happens because NLF(TS) is a Lyapunov function that exclusively trains to assess the stability of the test-time system instead of adapting to it. 

For our second and third systems, we take two power system examples with $N$ microgrids that use droop control schemes to stabilize the evolution of phase angles. The state $x$ captures modified phase angle $[\Delta \delta_1, \Delta \delta_2, \hdots, \Delta \delta_N]$, where $\Delta \delta_j = \delta_j- \delta_j^{*}$ denote the difference between the $j$-th phase angle and its setpoint. Among all parameters, the droop constants $(\text{dc}_1, \hdots, \text{dc}_N)$, which regulate the strength of subsystem-wise control loops, are chosen to be stochastic. We start off with a three-microgrid system with stochastic $\text{dc}_1$ and $\text{dc}_2$, and set $\vartheta_0 = (\text{dc}_1, \text{dc}_2, \text{dc}_3)_0 = (2.0, 2.0, 2.0)$ and $\vartheta_{n+1} = (3.0, 4.3, 2.0)$. As illustrated in Figure \ref{fig:ROAs-3mg-comp-1}, we find the ROA of the task-adapted NLF (derived from meta-NLF) to be larger than that of QLF(TS) and SOS-LF(TS) but smaller than NLF(TS)'s ROA, while T-NLF doesn't succeed in sketching an ROA. Next, we consider a five-microgrid system and make all droop constants stochastic. In this case, the nominal and test-time parameters are  $\vartheta_0 = (\text{dc}_1, \hdots, \text{dc}_5)_0 = (2.0, 2.0, 2.0, 2.0, 2.0)$ and  $(\text{dc}_1, \hdots, \text{dc}_5)_{n+1} = (3.5, 3.5, 3.0, 4.0, 3.2)$. Figure \ref{fig:ROAs-5mg-comp-5} compares the ROAs of all methods, and we notice the updated task-adapted NLF leaving behind QLF(TS), SOS-LF(TS), and T-NLF in terms of performance. Specifically, the latter two methods completely fail to produce ROAs. Finally, we hypothesize the following based on the preceding simulation results: (i) despite QLF and SOS-LF being easier to compute on the test-time system,  meta-NLF should be preferred over them in the interest of performance. (ii) T-NLF isn't a reliable adaptive method as its performance remains questionable when the parametric stochasticity tends to be high.

Our final system is the Caltech Ducted Fan in hover mode \cite{caltechfan}. The state variables are $x, y, \theta$ that denote the horizontal and vertical positions and angular orientation of the center of the fan. The system parameters include mass $(m)$, moment of inertia $(J)$, gravitational constant $(g)$, and coefficient of viscous friction $(d)$. We assume $m, r, d$ to be stochastic and fix $\vartheta_0 = (m, J, r, g, d)_0 = (11.2, 0.0462, 0.15, 0.28, 0.1)$ and $ \vartheta_0 = (13.0, 0.0462, 0.165, 0.28, 0.15)$. Figure \ref{fig:ROAs-cf-comp-2} compares all ROAs and conveys the message that the task-adapted NLF exceeds T-NLF and QLF(TS) in performance while SOS-LF(TS) method fails to converge. Similar to the previous experiments, we conclude here that meta-NLF outperforms T-NLF, QLF(TS), and SOS-LF(TS) while it stays competitive with NLF(TS).

 \begin{remark}
 For three-dimensional or larger systems taken in the experiments (in Figures \ref{fig:ROAs-3mg-comp-1}, \ref{fig:ROAs-5mg-comp-5} and \ref{fig:ROAs-cf-comp-2}), we show two-dimensional projections of ROAs onto specific planes for visual comparison purposes.
 \end{remark}
%%%%%%%%%%%%%%%%%%%%%%%%%%%%%%%%%%%%%%%%%%%%%%
%%%%%%%%%%%%%%%%%%%%%%%%% 6: Conclusion %%%%%%%%%
\section{Conclusion}
In this work, we presented a meta-learning-based Lyapunov function that adapts well to dynamical system instances under parametric shifts. To achieve this, meta-training is conducted using a set of dynamical system instances, where each instance represents a deterministic snapshot of a system exhibiting parametric stochasticity. Upon convergence of meta-training, the meta-function (which we term the meta-NLF) faces a new system instance and adapts to it quickly (in a few gradient steps). To the best of our knowledge, this is the first work that integrates meta-learning with NLFs to formulate adaptive stability certificates.
We compared meta-NLF with three non-adaptive baselines and an adaptive baseline method on different benchmark dynamical systems to highlight the superior adaptive performance of our method. One limitation of our present work is we assume a small dataset from the test-time system to be available on which our meta-function performs a few gradient update steps. In future work, we will extend our method to a meta-reinforcement learning framework and try to learn a policy function that adapts to an arbitrary test-time environment and stays both optimal and stable (in the Lyapunov sense) in it. Formulations of meta-learning-based barrier functions and contraction metrics are two other exciting research areas that we plan to study.

%%%%%%%%%%%%%%%%%%%%%%%%%%% Acknowledgements %%%%%

\section*{Acknowledgements}
This work is supported in part by NSF ECCS-2038963, the U.S. Department of Energy (DoE) Office of Energy Efficiency and Renewable Energy (EERE) under the Solar Energy Technologies Office (SETO) Award Number DEEE0009031, and Texas A\&M Engineering Experiment Station (TEES) Smart Grid Center.

%%%%%%%%%%%%%%%%%%%%%%%%%%%%%%%%%%%%%%%%%%%%%

\bibliography{aaai24}
\onecolumn
\appendix
%%%%%%%%%%%%%%%%%%%%%%%%%%
\section{Verification of $V_\theta(x)$ being Positive Definite in $D$}
In the subsection \textit{Verification of Meta-NLF}, we informally propose that (i) if a Lipschitz task-adapted Lyapunov function $V_{\theta^{'}}(.)$ satisfies a tightened version of Lyapunov constraint \eqref{eq:lyapunov-condition-2} everywhere on a discretized grid, and (ii) the discretized grid tightly covers $D$, then $V_{\theta^{'}}(.)$ stays positive definite everywhere in $D$. Here, we formalize this hypothesis and provide its proof. Before proceeding to the proposition, we define $D_\tau$ as a discretized grid of valid region $D$ such that that $\|x-[x]_\tau\|_1 \leq \tau$ for any $x \in D$, where $[x]_\tau$ is the nodal point in $D_\tau$ that has the smallest $\ell_1$ distance to $x$.

\begin{proposition}
    Let Assumption 1, equation \eqref{eq:-lip-lyap-a} hold. Then, an arbitrary task-adapted NLF $V_{\theta^{'}}(.)$ satisfies $V_{\theta^{'}}(x) > 0$, $x\in D$, if the following holds:
    \begin{equation}
    V_{\theta^{'}}([x]_\tau) > K_V \tau \qquad \forall [x]_{\tau} \in D_\tau,
    \end{equation}
    where $K_V$ is the Lipschitz constant and $\tau > 0$ signifies how densely $D_\tau$ covers $D$.
\end{proposition}
\begin{proof}
    Using the Lispschitzness assumption of $V_{\theta^{'}}$ (equation \eqref{eq:-lip-lyap-a}) and definition of $D_\tau$, we have the following for any $x \in D$:
    \begin{eqnarray}
        & |V_{\theta^{'}}(x) - V_{\theta^{'}}([x]_\tau)|
        \leq K_V\|x - [x]_\tau\|_1 \leq K_V \tau \nonumber \\
        \implies & V_{\theta^{'}}([x]_\tau)-K_V \tau\leq V_{\theta^{'}}(x) \leq V_{\theta^{'}}([x]_\tau) +K_V \tau.
    \end{eqnarray}
Using the above inequality, we have $V_{\theta^{'}}(x) > 0 $ whenever $V_\theta([x]_\tau) > K_V \tau$. This completes the proof.
\end{proof}
Finally, by choosing $\epsilon_1 = K_V \tau$, we can provide a tightened condition $V_{\theta^{'}}([x]_\tau) > \epsilon_1$, for $[x]_\tau \in D_\tau$,  which will lead to the satisfaction of $V_{\theta^{'}}(x) >0$, for $x \in D$.

\section{Additional Experimental Results}
All experiments have been conducted in a 3.0GHz, 6-core Intel Corp i7-9700K CPU. We have used Python-based Numpy and Pytorch packages to train Meta-NLF and all other neural network-based baseline methods. For Quadratic and SOS Lyapunov functions, we have utilized the Control System Toolbox on Matlab. The additional details on the baseline methods and dynamical systems are detailed below.

\subsection{Baseline Methods}
In \textit{Experiments} section, we have employed SOS-LF(TS), QLF(TS), NLF(TS), and T-NLF as baseline methods. We give the details of these methods below.
\begin{enumerate}
    \item \textbf{Sum of Squares-based Lyapunov function (SOS-LF)}: Traditionally, SOS programming is utilized to obtain the Lyapunov function of a system with polynomial dynamics. However, a dynamical system may contain sinusoidal or exponential terms for which the SOS method isn't applicable. As suggested in \cite{richards2018lyapunov}, we replace $\sin(x)$, $\cos(x)$, and all other non-polynomial terms in the dynamics with their Taylor series expansions around the origin. This way, non-polynomial dynamics can be transformed into a polynomial one on which the SOS method can be implemented.

    \item \textbf{Quadratic Lyapunov function (QLF)}: This approach linearizes the system dynamics around the origin and obtains a quadratic function that satisfies the Lyapunov constraints on the linearized system. In essence, this function serves as a local Lyapunov function that assesses the stability of the non-linear system in a small neighborhood around the origin.
    \item \textbf{Neural Lyapunov function (NLF)}: This is a data-driven method that uses a neural network to approximate a Lyapunov function on any given system dynamics \cite{chang2019neural}. Based on a training dataset containing randomly-sampled state vectors and  the dynamics $f(x)$ evaluated at these samples, NLF training utilizes gradient descent steps to optimize a loss function that penalizes the violation of Lyapunov constraints. This work has been initially proposed in \cite{chang2019neural}, where an additional SMT solver numerically checks if the trained neural network violates the Lyapunov constraints inside a given valid region in a state space. We have adopted an alternative Lipschitz approach \cite{richards2018lyapunov} in the experiments because  SMT checkers don't scale well to larger systems and often completely fail to converge to a solution.
    \item \textbf{Transfer learning-based neural Lyapunov function (T-NLF)}: We modify the standard NLF by fully training it on the nominal system and then transferring it to the test-time system with a few gradient updates. We note that during testing, we have limited data and time available to update the NLF on the test-time system. Same as in the previous algorithm, we don't use the SMT solver because of its lack of scalability.
\end{enumerate}
\subsection{Dynamical Systems}
\subsubsection{Inverted Pendulum}
\label{app:IP}
The Inverted Pendulum is a setup with mass $m$ attached to a hinge with a coefficient of friction $b$ via a massless rod of length $l$. A torque $u$ acts as a control input and tries to keep the pendulum upright, which otherwise will fall. The system dynamics are described as the following:
\begin{equation}
    \ddot \theta = \dfrac{mgl\sin{\theta} - b\dot\theta + u}{ml^2},
\end{equation}
where $\theta$ is the angular displacement from the upright position and  $\dot \theta$ is the angular velocity. We use an LQR-based controller to compute $u$ and make the system a closed-loop one.

\subsubsection{A power system containing $N$ microgrids}
Power systems are networks of electrical components designed to generate, supply, and utilize electric power. Most power systems can be treated as autonomous control systems that use droop control schemes for voltage, phase angle, or frequency stabilization. We consider a power system network containing $N$ microgrids that are governed by the following dynamics:
\begin{eqnarray} 
    \label{eq:MG_interface_dynamics}
    & J_{\delta_i} \Delta\dot{\delta_i} = -\text{dc}_i\Delta \delta_i-\Delta P_i +K_i\Delta\delta_i \\ 
    & P_i = \sum_{k\ne i}E_i^{*}E_k^{*}Y_{ik} \cos(\delta_{ik} - \gamma_{ik}) \nonumber \\
    & + {E_i^{*}}^2G_{ii}, i\in \{1,\hdots, N\},
\end{eqnarray}
where  $\Delta \delta_i =\delta_i-\delta_i^*$; $\Delta P_i =P_i-P_i^*$, $\delta_i^*$, $P_i^*$, $Q_i^*$ are the nominal set point values of the voltage magnitude, phase angle, and active power at the $i$\text{th} microgrid, respectively; $J_{\delta_i}$s are tracking time constants; $\text{dc}_i$s are droop coefficients; $K_i$ is the $i-$th entry in an output-feedback gain matrix and $Y_{ik}$ and $G_{ii}$ are Admittance matrix parameters. This is an autonomous system functioning with droop controllers.

\subsubsection{Caltech Fan in hover mode}
Caltech Ducted Fan is another standard control system example developed by Cal Tech that mimics a simplified flight control experiment \cite{caltechfan}. The state variables are $x, y, \theta$ that denote the horizontal and vertical positions and angular orientation of the center of the fan. There are two restoring forces $f_1$ and $f_2$, where $f_1$ acts perpendicular to the axis at a distance $r$, and $f_2$ acts parallel to the fan's axis. To bring the equilibrium point to the origin when inputs are zero, the control inputs are defined as $u_1 = f_1$ and $u_2 = f_2 - mg$.  The system parameters are mass $(m)$, moment of inertia $(J)$, gravitational constant $(g)$, and coefficient of viscous friction $(d)$. The dynamics of this system are given as follows:
\begin{eqnarray}
& m \ddot{x} =-m g \sin \theta-d \dot{x}+u_1 \cos \theta-u_2 \sin \theta \\
& m \ddot{y} =m g(\cos \theta-1)-d \dot{y}+u_1 \sin \theta+u_2 \cos \theta \\
& J \ddot{\theta}  =r u_1.
\end{eqnarray}
We incorporate an LQR controller to close the control loop which helps make the system autonomous.  

\subsection{Additional Experiments}
% \textcolor{red}{-----marker-------}
Here, we present additional experiments to compare the stability assessment performance of Task-adapted NLF with all baselines. Next, as discussed in \textit{Experiments} section, we create a nominal and a test-time system for each experiment to test out the performance of each method. We note that the default parameter $\vartheta_{0}$ has been held fixed across all experiments performed on a specific system. However, the test-time parameters $\vartheta_{n+1}$ differ for each experiment.

\subsubsection{Additional Experiments on Inverted Pendulum}
We recall that the system parameter vector is $\vartheta = (l,m,g,b)$. The default parameter for all experiments is $\vartheta_0 = (0.5, 0.15, 9.81, 0.1)$. First, parameters $l, b$ are considered uncertain, and the test-time parameter $\vartheta_{n+1}$ is set as $(1.35, 0.15, 9.81, 0.2)$. In the following experiment, all four parameters, i.e., $l, m, g, b$, are assumed to be stochastic, and $\vartheta_{n+1}$ is sampled to be $(1.0, 0.2, 9.0, 0.2)$. Figures \ref{fig:ROAs-IP-comp-2} and \ref{fig:ROAs-IP-comp-3} illustrate the ROA comparisons for the first and second cases, respectively. In both cases, task-adapted NLF (obtained from adaptation of meta-NLF) is observed to (i) outperform T-NLF, QLF(TS) and SOS-LF(TS), (ii) and closely compete with NLF(TS) in performance. The ROAs are compared quantitatively for each experiment in Table \ref{ROA-table}.

\begin{figure}[t]
       \begin{subfigure}[b]{\textwidth}
         \centering
         \includegraphics[width=\textwidth]{plots/IP/Legend.pdf}
     \end{subfigure}
    \begin{subfigure}[b]{0.45\textwidth}
         \centering
         \includegraphics[width=\textwidth]{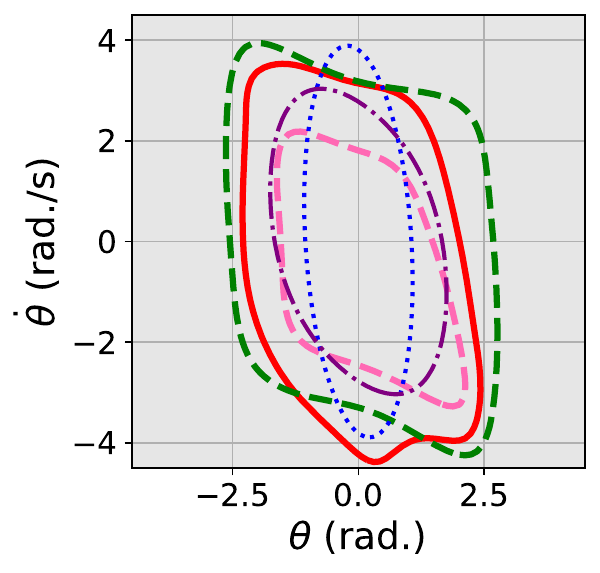}
         \caption{}
          \label{fig:ROAs-IP-comp-2}
     \end{subfigure}
     \hfill
     \begin{subfigure}[b]{0.45\textwidth}
         \centering
         \includegraphics[width=\textwidth]{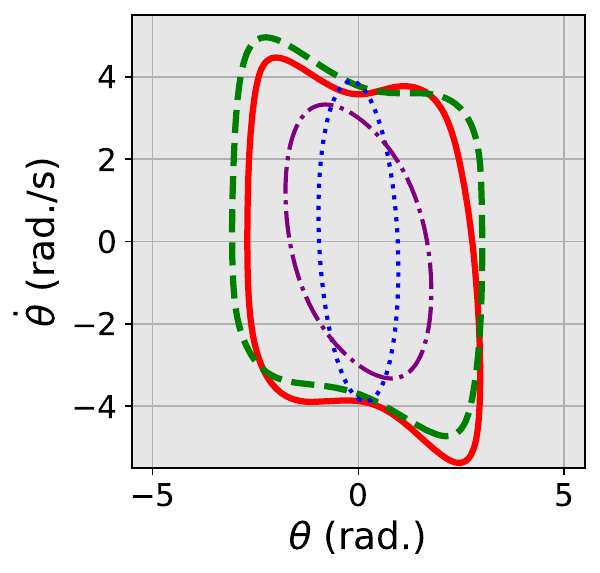}
        \caption{}
        \label{fig:ROAs-IP-comp-3}
     \end{subfigure}
     \caption{\textit{ROA Comparison plots for Inverted Pendulum}. For an Inverted Pendulum setup, two different test cases are created by making (a) stochastic $(l, b)$ (b) stochastic $(l, m, g, b)$. In each case, the non-stochastic parameters are set at nominal values. Task-adapted NLF (obtained through test-time adaptation of meta-NLF) is better adaptive than SOS-LF(TS) and QLF(TS), as the former achieves a larger ROA. T-NLF fails to capture an ROA in (b), which shows that a naive way of transferring an NLF to the test-time system followed by a fine-tuning step doesn't always result in a valid Lyapunov function.}
    \end{figure}

\subsubsection{Additional Experiments on N-microgrid Systems}
We start with a three-microgrid system whose droop constants are assumed to be stochastic parameters that randomly vary over time. The nominal parameter is set as the following, $\vartheta_0 = (\text{dc}_1, \text{dc}_2, \text{dc}_3)_0 = (2.0, 2.0, 2.0)$. In the additional experiment, we assume all droop constants are stochastic, and the test-time parameter $\vartheta_{n+1}$ is $(2.5, 4.5, 3.9)$. The ROAs are compared in Figure \ref{fig:ROAs-3mg-comp-2}. Next, we take a five-microgrid system (illustrated in Figure \ref{fig:IEEE-test-feeder}), and similar to the previous setup, the droop constants are regarded as uncertain. The nominal and test-time parameters are $\vartheta_0 = (\text{dc}_1, \text{dc}_2, \text{dc}_3, \text{dc}_4, \text{dc}_5)_0 = (2.0, 2.0, 2.0, 2.0, 2.0)$ and $\vartheta_{n+1} = (3.5, 3.5, 3.0, 4.0, 3.2)$. In Figure \ref{fig:ROAs-5mg-comp-4}, we compare ROAs of all stability assessment methods. Both figures concurrently convey that, the task-adapted NLF's performance is better than T-NLF, QLF(TS), and SOS-LF(TS). Also, task-adapted NLF is almost as good as NLF(TS) on the test-time system which justifies our claim for task-adapted NLF's superior adaptive performance. The ROAs are compared quantitatively in Table \ref{ROA-table}

We note that in each microgrid-based test case above, except the droop constants, all other networked microgrid parameters have been taken from \cite{jena2022distributed}. 

\begin{figure}[t]
       \begin{subfigure}[b]{\textwidth}
         \centering
         \includegraphics[width=\textwidth]{plots/IP/Legend.pdf}
     \end{subfigure}
    \begin{subfigure}[b]{0.3\textwidth}
         \centering
         \includegraphics[width=\textwidth]{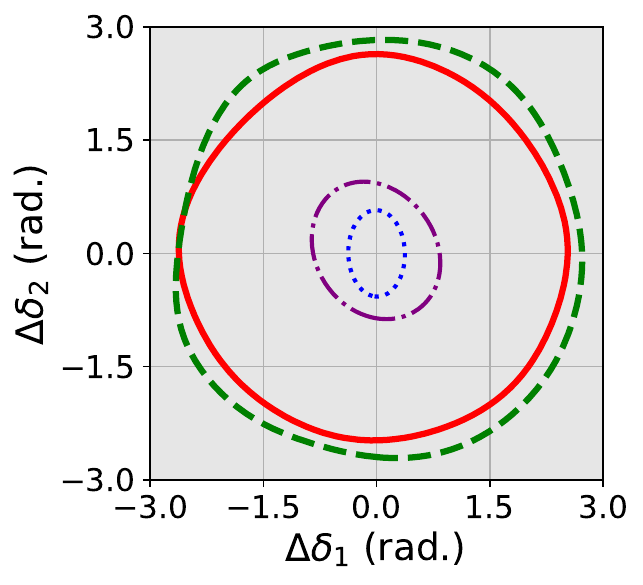}
       \caption{}
        \label{fig:ROAs-3mg-comp-2}
     \end{subfigure}
     \hfill
     \begin{subfigure}[b]{0.3\textwidth}
         \centering
         \includegraphics[width=\textwidth]{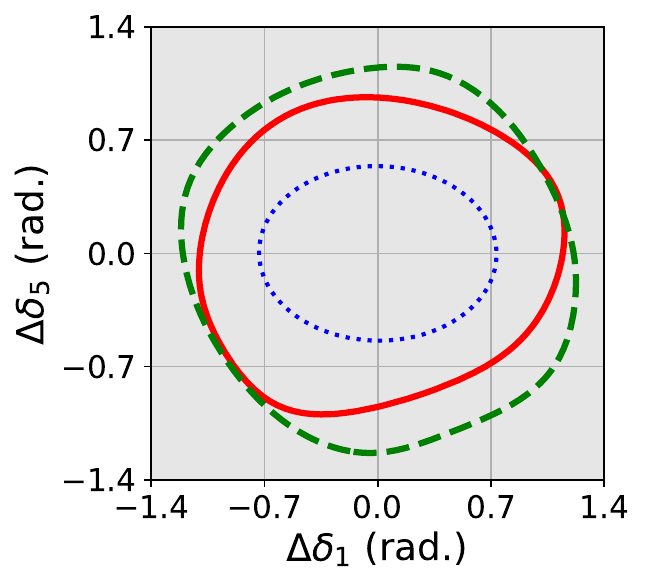}
      \caption{}
        \label{fig:ROAs-5mg-comp-4}
        \end{subfigure}
    \hfill
    \begin{subfigure}[b]{0.3\textwidth}
         \centering
         \includegraphics[width=\textwidth]{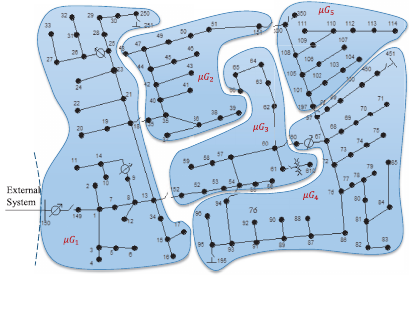}
      \caption{}
        \label{fig:IEEE-test-feeder}
     \end{subfigure}
     \caption{\textit{ROA Comparison plots for power systems with $N$ microgrids.}  We have two test cases, i.e., (a) a 3-microgrid system with stochastic droop constants $(\text{dc}_1, \text{dc}_2, \text{dc}_3)$ ;(b) A 5-microgrid system with stochastic droop constants $(\text{dc}_1, \text{dc}_2)$. We present an IEEE 123-node test feeder network partitioned into a 5-microgrid system (figure taken from \cite{jena2022distributed}) in (c). The task-adapted NLF (computed by updating the meta-NLF on the test-time system) is almost as optimal as NLF(TS) and exceeds all other benchmark methods in both cases. The SOS-LF(TS) doesn't converge for case (b) and thus isn't included in the corresponding figure.}
    \end{figure}

\subsubsection{Additional Experiments on Caltech Fan in Hover Mode}
A schematic diagram of Caltech Fan is shown in 
Figure \ref{fig:caltech_fan}. Now, we recall the system parameters to be $(m, J, r, g, d)$ whose nominal value is set as $(11.2, 0.0462, 0.15, 0.28, 0.1)$. We make $m$ stochastic and sample a test-time parameter as $(13.0, 0.0462, 0.15, 0.28, 0.1)$. The ROAs of all methods are presented in Figure \ref{fig:ROAs-cf-comp-1}, where task-adapted NLF, like all previous experiments, exceeds QLF(TS), SOS-LF(TS), and T-NLF and performs nearly as optimal as NLF(TS) on the test-time system.

\begin{figure}[t]
       \begin{subfigure}[b]{\textwidth}
         \centering
         \includegraphics[width=\textwidth]{plots/IP/Legend.pdf}
     \end{subfigure}
         \centering
     \begin{subfigure}[b]{0.45\textwidth}
         \centering
         \includegraphics[width=\textwidth]{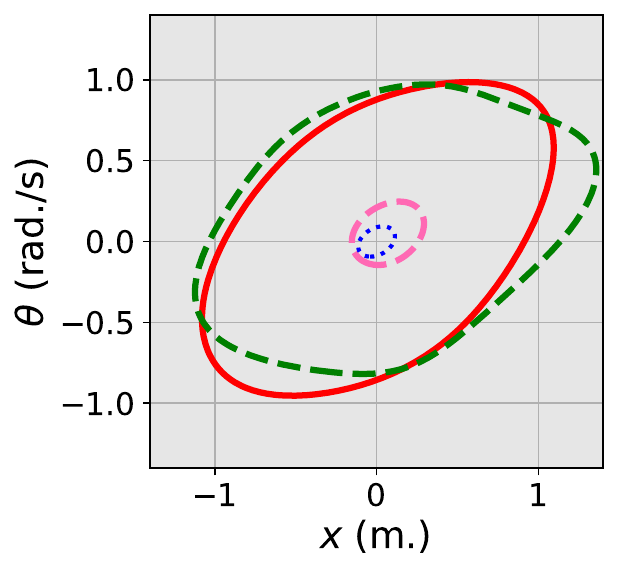}
           \caption{}
        \label{fig:ROAs-cf-comp-1}
     \end{subfigure}
     \hfill
     \begin{subfigure}[b]{0.45\textwidth}
         \centering
         \includegraphics[width=\textwidth]{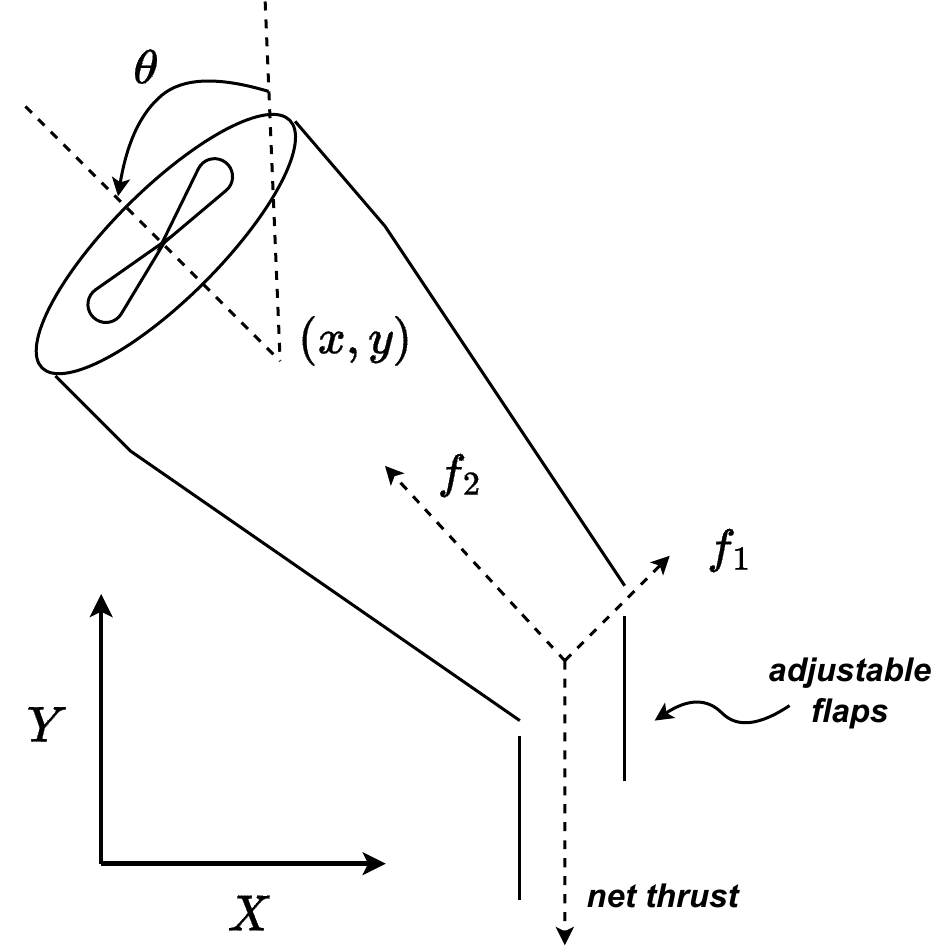}
          \caption{}
        \label{fig:caltech_fan}
     \end{subfigure}
     \caption{\textit{ROA Comparison plots for the Caltech Fan in hover mode.} The test case is built by assuming a stochastic $m$ in (a);  In (b), we have a diagrammatic representation of the Caltech Fan. The ROA of task-adapted NLF (computed by meta-NLF's test-time adaptation) is close to the ROA of NLF (TS) and exceeds that of all other baseline methods. SOS-LF(TS) algorithm fails to converge and hence is omitted.}
    \end{figure}
 %%% test %%%%%%%%%

\begin{table*}[b]
\caption{A Quantitative Comparison of Stability Assessment Performances of Task-adapted NLF and Baseline \\ Methods}
\label{ROA-table}
\vskip 0.15in
\begin{center}
\begin{small}
\begin{sc}
\begin{tabular}{lcccccr}
\hline
Env. &  Stoch. params & NLF (TS) & QLF (TS) & SOS-LF (TS) & T-NLF & Task-adapted\\
&&&&&&NLF\\[1ex]
\hline
Adaptive ? & & \xmark  & \xmark & \xmark  & \cmark & \cmark \\[1ex]
\hline
IP    & $(l)$      & 31.83     & 12.42     & 17.96   & 0 & 15.04\\[1ex]
    & $(l , b)$      & 34.28   & 12.90     & 15.78     & 13.72 & 28.04\\[1ex]
    & $(l, m, g, b)$ & 46.04   & 11.67     & 17.56   & 0 & 42.67\\[1ex]
    \hline
    3-MG   & $(dc_1, dc_2)$    & 71.94       & 1.60    & 6.82     & 0 & 52.89 \\[1ex]
       & $(dc_1, dc_2, dc_3)$ & 66.99            & 2.47    & 5.63    & 0  & 60.41 \\[1ex]
    \hline
5-MG   & $(dc_1, dc_2)$    & 42.83       & 13.02    & -    & 0 & 37.09\\ [1ex]
    & $(dc_1, dc_2, dc_3, dc_4, dc_5)$    & 58.93       & 17.10    & -    & 0  & 44.35\\ [1ex]
    \hline
CF  & $(m)$    & 33.35      & 0.23         & -       & 1.5  & 25.16 \\ [1ex]
    & $(m, r, d)$          & 21.45     & 0.22        & -       & 1.43   & 13.72 \\
\hline
\end{tabular}
\end{sc}
\end{small}
\end{center}
\vskip -0.1in
\end{table*}

\end{document}